\documentclass[letterpaper, 10 pt, conference]{ieeeconf}  % Comment this line out if you need a4paper

\IEEEoverridecommandlockouts                              % This command is only needed if 
\usepackage{cite}

\usepackage{amsmath,amssymb,amsfonts}
\usepackage{amsthm}

\usepackage{algorithmic}
\usepackage{graphicx}
\usepackage{textcomp}
\usepackage{xcolor}
\newtheorem{prop}{Proposition}
\newtheorem{lemma}{Lemma}
\newtheorem{assumption}{Assumption}
\newtheorem{theo}{Theorem}

\newtheorem{defi}{Definition}
\newtheorem{exmp}{Example}[section]
\usepackage[norelsize, linesnumbered, ruled, lined, boxed, commentsnumbered]{algorithm2e}

\usepackage[normalem]{ulem}
\usepackage{mathabx}

\title{\LARGE \bf
Structural Stability of a Family of \\Group Formation Games 
}
\author{Chenlan Wang$^{1}$, Mehrdad Moharrami$^{2}$, Kun Jin$^{1}$, David Kempe$^{3}$, P. Jeffrey Brantingham$^{4}$,  and Mingyan Liu$^{1}$% <-this % stops a space
\thanks{$^{1}$ University of Michigan, Ann Arbor, USA}%
\thanks{$^{2}$ University of Illinois at Urbana-Champaign, Illinois, USA}%
\thanks{$^{3}$ University of Southern California, Los Angeles, USA}%
\thanks{$^{4}$ University of California, Los Angeles, Los Angeles, USA}%
}

\begin{document}

\maketitle
\thispagestyle{empty}
\pagestyle{empty}

%%%%%%%%%%%%%%%%%%%%%%%%%%%%%%%%%%%%%%%%%%%%%%%%%%%%%%%%%%%%%%%%%%%%%%%%%%%%%%%%
\begin{abstract}
We introduce and study a group formation game in which individuals/agents, driven by self-interest, team up in disjoint groups so as to be in groups of high collective strength. This strength could be group identity, reputation, or protection, and is equally shared by all group members. The group's access to resources, obtained from its members, is traded off against the geographic dispersion of the group: spread-out groups are more costly to maintain. We seek to understand the stability and structure of such partitions. We define two types of equilibria: 
\begin{enumerate}
	\item Acceptance Equilibria (AE), in which no agent will unilaterally change group affiliation, either because the agent cannot increase her utility by switching, or because the intended receiving group is unwilling to accept her (i.e., the utility of existing members would decrease if she joined); and 
	\item Strong Acceptance Equilibria (SAE), in which no subset of any group will change group affiliations (move together) for the same reasons given above. 
\end{enumerate} 
We show that under natural assumptions on the group utility functions, both an AE and SAE always exist, and that any sequence of improving deviations by agents (resp., subsets of agents in the same group) converges to an AE (resp., SAE). We then characterize the properties of the AEs. We show that an ``encroachment'' relationship --- which groups have members in the territory of other groups --- always gives rise to a directed acyclic graph (DAG); conversely, given any DAG, we can construct a game with suitable conditions on the utility function that has an AE with the encroachment structure specified by the given graph. 
\end{abstract}

%%%%%%%%%%%%%%%%%%%%%%%%%%%%%%%%%%%%%%%%%%%%%%%%%%%%%%%%%%%%%%%%%%%%%%%%%%%%%%%%
\section{Introduction}
The formation of groups and group membership plays an important role in human societies. Individuals form groups (or join existing groups) to benefit from shared interests or resources, reputation, protection, safety, monetary rewards, etc. For example, clubs are formed by individuals sharing similar interests, gangs form (among others) to provide their members protection \cite{gravel2018birds},
and states form strategic alliances. The process by which groups are formed, and their resulting stability, is of great interest and has been studied in many fields, including computer networks, social sciences, economics, and political sciences \cite{goette2006impact,milchtaich2002stability,hollard2000existence}. Different models have been introduced to study such processes; including coalition formation games \cite{slikker2001coalition,bogomolnaia2002stability,banerjee2001core,ray2015coalition,aziz2013computing,aziz2011stable,igarashi2016hedonic,igarashi2017coalition,hoefer2018dynamics,hagen2020two,dreze1980hedonic}, agent-based modeling \cite{collins2017agent,collins2018strategic}, and signed network formation \cite{hiller2017friends}.

One of the key tradeoffs in the formation of groups is between adding more resources to a group vs.~making it less cohesive. In the examples above, when a club broadens its interests, it can add more members, but at the cost of less thematic cohesion. When a gang expands its territory, it can add more members (who add the ability to offer protection to others), but at the cost of less spatial cohesion, which makes it harder to coordinate actions.
Similarly, when a strategic alliance comprises more countries, it can draw on more strength, but will suffer in less geographic cohesion. The big-picture question we investigate in this work is the following: \emph{what are the effects of this tradeoff in terms of the group structures that one observes at equilibrium?}

To answer this question more formally, we introduce a group formation game (GFG) in which the group members' utilities capture a trade-off between the combined resources of the group and the geographical dispersion. At a high level, our GFG is similar to hedonic games \cite{bogomolnaia2002stability} studied within the family of coalition games. In the GFG we define (formal definitions are in Section~\ref{sec:group formation game}), agents exist in a Euclidean space. Each agent is endowed with a real-valued amount of individual resources. Groups form endogenously, and agents will form/join groups so as to maximize their utilities. Agents derive their utility from the utility of the group they belong to, plus optionally an individual component. The group utility is the same for all group members and is a linear combination of two terms: the power of the group which only depends on the group's members (their resources and locations); and the negative effect of more powerful groups which is a linear function of their group members (resources and locations). In Section~\ref{sec: Structural Properties} where we study structural properties of the groups that are formed at equilibrium, we introduce a slightly stronger assumption that group utility depends only on its own members (their resources and locations), but not on any other group. This latter assumption is similar to hedonic preferences, which also ignore dependencies between coalitions \cite{dreze1980hedonic}, although in a hedonic game, typically, only preference orders are specified, without numerical utility functions. The utility of a group is a monotone increasing function of the total resources of the group's members, and a monotone decreasing function of how ``dispersed'' the group is. We define and study several such notions, specifically, the volume, diameter, and surface area of the group.

The groups form endogenously, with each agent choosing a strategy in a self-interested way.
In choosing an appropriate equilibrium notion, we note that the settings described above typically allow agents to join a group only with the group's consent (whereas individuals can of course leave a group unilaterally). We therefore define two notions of equilibria/stability that generalize the pairwise equilibria of \cite{jackson1996strategic}; the first is similar to the {\em individual stable equilibrium}  studied by \cite{dreze1980hedonic} in the context of hedonic games.\footnote{
	%The main difference lies in the ``strict benefit'' required in our game: an agent can only be accepted if this strictly increases the utility of those in a receiving group, whereas in \cite{dreze1980hedonic}, an agent can join a group if this does not lower the utility of those already in the group.
	This equilibrium definition is equivalent to the one in \cite{dreze1980hedonic};  however, our game is different from the game in \cite{dreze1980hedonic}.}
\begin{enumerate}
	\item Acceptance Equilibrium (AE): each agent (weakly) prefers membership in her group over any group that would (weakly) prefer the agent to join that group.%
	\footnote{This of course includes that the agent would not prefer to form a new (singleton) group by herself.}\item Strong Acceptance Equilibrium (SAE): each subset of agents who are currently in the same group (weakly) prefers membership in their current group over joint membership in any group that would (weakly)  prefer the entire subset to join.\footnote{Again, this includes that this subset of agents would not prefer to form a new group by themselves.} 
\end{enumerate}

Several facts are worth noting about the notions of AE and SAE:
(i) The notion of AE is different from the standard definition of Nash equilibrium \cite{bogomolnaia2002stability}, in which agents can join a group unilaterally, without the group's approval.
(ii) This type of approval by groups is a natural requirement in many realistic scenarios; see, e.g., \cite{gravel2018birds}. (iii) An alternative characterization of AE is obtained by considering a normal-form representation of the game. Each agent chooses a subset of agents including herself as a strategy. When these strategies are consistent with each other (i.e., the sets form a partition into ``cliques''), the result is the groups given by the strategies, with corresponding payoffs\footnote{Under this characterization, choosing a subset to form a group with is perhaps more aptly considered a ``preference'' rather than a ``strategy'' or ``action''; for simplicity, we will use these terms interchangeably in this paper.}. Otherwise, the result has utility $-\infty$ for all agents. Then, SAE can be viewed as equilibria under possible group deviations. 

We show (in Section~\ref{sec:existence of equilibrium}) that under natural assumptions on the group utility functions, both an AE and SAE always exist; this is done by showing that a better-response type of algorithm always finds such an equilibrium. Under a certain condition, we also introduce a polynomial time algorithm that characterizes an SAE.
%Similarly, we show the existence of a particular type of SAE using an algorithm, that has polynomial complexity if diameter is used to measure geographical dispersion of groups. 
These equilibria are in general not unique.
Our main focus (in Section~\ref{sec: Structural Properties}) is then on structural properties of the groups that form at equilibrium.  We define the territory of a group as the convex hull of the individuals inside the group; our main focus is on understanding the relationship between the convex hulls of different groups, and to what extent they may overlap or contain each other, for different families of group utility functions. A particularly useful tool in characterizing the group structure is a graph capturing the territorial relationships among groups. We explore which graphs can be obtained at equilibria of the game in this way.

%\dk{[DK: I'm marking replacement suggestions with ``R:'', deletion
 % suggestions by putting them inside square brackets, and additions by just putting them there.]}
%\textcolor{cyan}{
A practical application of our model is to infer the power relation between gang groups from their structures at equilibrium. For example, 
%\dk{[considering gang groups]}
if one gang
%\dk{gang} 
group is nested within another group, and this territorial relation is stable, then it indicates that the inner group is more powerful than the outer group; otherwise the inner group would be absorbed by the outer group.
Our model can be applied to predict stable group structures given individual locations and resources. 
%Last but not least, our model can be applied beyond physical groups such as online social groups in which individuals have virtual rather than physical locations. 
%\dk{[DK: Do we really want to include this? Why would these coordinates be Euclidean (which is the only case we handle)? In what sense would geographic dispersion matter? If we think that we have something important and meaningful to say here, we should elaborate, but otherwise, perhaps just delete this sentence?]}
%}
%Besides, our model can also be applied to predict stable group structures when the initial characters of individuals are given. For instance, given individual locations and resources, we can predict the stable groupings that each individual could prefer.   }

The literature most relevant to our work is on Coalition Formation Games, which have been studied in the context of stable marriage, roommate assignment, and research team formation; see, e.g., \cite{hagen2020two,hoefer2018dynamics,igarashi2017coalition,igarashi2016hedonic,ui2000shapley,slikker2001coalition,aziz2013computing,aziz2011stable}. The game we study is closest to hedonic games \cite{bogomolnaia2002stability,banerjee2001core}, in which groups are formed based on individual declared preferences and preference rankings, often in the absence of specific utility functions. By contrast, our game is first and foremost defined by individual utility functions, where preferences serve as a convenient and alternative way of representing strategy profiles. In addition, our game is not limited to hedonic settings. Beyond coalition formation games, group formation has been studied computationally using Agent-Based Models (ABM) \cite{collins2017agent,collins2018strategic}, in the context of group identity \cite{dev2018group}, using agent similarity \cite{milchtaich2002stability}, and by modeling inter-group conflict \cite{hiller2017friends,heider1946attitudes, cartwright1956structural, davis1967clustering}.  
%%%%%%%%%%%%%%%%%%%%%%%%%%%%%%%%%%%%%%%%%%%%%%%%%%%%%%%%%%%%%%%%%%%%%%%%
%%%%%%%%%%%%%%%%%%%%%%% SEC2: GROUP FORMATION GAME %%%%%%%%%%%%%%%%%%%%%
%%%%%%%%%%%%%%%%%%%%%%%%%%%%%%%%%%%%%%%%%%%%%%%%%%%%%%%%%%%%%%%%%%%%%%%%
\section{The Group Formation Game (GFG)}
\label{sec:group formation game}

%%%%%%%%%%%%%%%%%%%%%%%%%%%%%%%%%%%%%%%%%%%%%%%%%%%%%%%%%%%%%%%%%%%%%%%%
%%%%%%%%%%%%%%%%%%%%% Subsec: formulation %%%%%%%%%%%%%%%%%%%%%%%%%%%%%%
\subsection{The model} \label{sec:model}
Each of the $n \geq 2$ individuals/agents is indexed by $i \in \mathcal{N}=\{1,2,...,n\}$.
Agent $i$ is located at $x_i \in \mathbb{R}^d$ (where $d \geq 1$ is the dimension of the Euclidean space), and has positive scalar resources (abilities, skills, characteristics, etc.) $r_i > 0$.
We write $\mathbf{x} = (x_1, \ldots, x_n)$ for the vector (technically, matrix) of all agents' locations, and $\mathbf{r} = (r_1, \ldots, r_n)$ for the vector of all agents' resources.

We study group formation as a one-shot game in which each individual chooses the agents with whom she wants to be in the same group.
Thus, agent $i$'s strategy/action/preference space is $\mathcal{A}_i = \{ A \subseteq \mathcal{N}|i \in A\}$;
we will denote agent $i$'s action by $a_i \in \mathcal{A}_i$.
The set of all joint action profiles is $\mathcal{A} = \bigtimes_{i=1}^n \mathcal{A}_i$, and a joint action profile is denoted by $\mathbf{a} = (a_1, a_2, \ldots, a_n)$.
The profile of the actions of all agents except $i$ is $\mathbf{a}_{-i}$.
In particular, we can write $\mathbf{a} = (a_i, \mathbf{a}_{-i})$.

We are interested only in profiles under which the groups form a \emph{disjoint partition} of the individuals, and in which the actions chosen by agents are consistent.
%We say that the action profile $\mathbf{a}$ is \emph{feasible} if and only if \textcolor{cyan}{it holds for every $i \in \mathcal{N}: j \in a_i \Rightarrow i \in a_j$}; %$a_i = \{j \in \mathcal{N}|a_j = a_i\}$ for all $i\in\mathcal{N}$;
We say that the action profile $\mathbf{a}$ is \emph{feasible} if and only if $a_i = \{j \in \mathcal{N}|a_j = a_i\}$ for all $i\in\mathcal{N}$;
in words, if all agents that $i$ wants to be in a group with also want to be in the \emph{same} group.

A feasible profile $\mathbf{a} = (a_1, a_2, \ldots, a_n)$ partitions the players $\mathcal{N}$ into $m = m(\mathbf{a}) \leq n$ disjoint \emph{groups}.
We index these groups as $G_1(\mathbf{a}), G_2(\mathbf{a}), \ldots, G_m(\mathbf{a})$; we will discuss a specific useful indexing scheme in Section \ref{sec:ordering}.
In this and other notation, we omit the dependence on $\mathbf{a}$ for legibility when it is clear from the context.
When the preference profile $\mathbf{a}$ is feasible, and the indexing of resulting groups has been fixed, we use $\sigma_i = \sigma_i(\mathbf{a})$ to denote the index of the (unique) group $k$ such that $i \in G_k(\mathbf{a})$. We use the same notation for subsets of a group, i.e., for $S\subset G_k(\mathbf{a})$ we set $\mathbf{\sigma}_S = k$.\footnote{If the members of $S$ are not all in the same group, then $\mathbf{\sigma}_S$ is undefined.} In that case (when the groups are clear from context), we often consider $\mathbf{\sigma} = \mathbf{\sigma}(\mathbf{a})$ itself to be the strategy (or group affiliation) profile.

As discussed in the introduction, the reasons motivating individuals to form groups include seeking protection, pooling resources, and gaining reputation, among others.  
Within this context, we 
introduce two assumptions on the individual utility functions. 

\begin{assumption} \label{assp:group-depends-on-members}
	The utility of a group depends on its own members' resources
	and locations as well as on the composition of other groups as follows: 
	\begin{align}\label{eq:group-utility}
	U_{G}(\mathbf{r}, \mathbf{x}) & = P_{G}(\mathbf{r}, \mathbf{x}) - \sum_{G': P_{G'}(\mathbf{r}, \mathbf{x}) > P_G(\mathbf{r}, \mathbf{x}) } H_{G'}(\mathbf{r}, \mathbf{x}),
	\end{align}
	where $P_G(\mathbf{r}, \mathbf{x})$ is a function that represents power of group $G$ and depends only on $\{r_i|i \in G\}$ and $\{x_i|i \in G\}$, and $H_{G'}(\mathbf{r}, \mathbf{x})$ is a non-negative function of $\{r_i|i \in G'\}$ and $\{x_i|i \in G'\}$, for groups $G'$ that are more powerful than $G$. 
	%In particular, it does not depend on less powerful groups.
\end{assumption}

This assumption applies in situations where less powerful groups are negatively influenced by more powerful groups. 
If $H_{G'}(\mathbf{r}, \mathbf{x}) = 0$, then the groups' utilities are independent of one another and depend only on their own membership; thus, the hedonic setting is considered as a special case. In the hedonic setting, group utility is equivalent to group power. %\textcolor{cyan}{
	Notice that this assumption is general in the sense that the power of each group, i.e., $P_G(\mathbf{r}, \mathbf{x})$, depends only on its members' locations and resources, and there are no restrictions on the forms of those quantities.
	%}
	
	% Under this general assumption and the second assumption which is introduced next, we show the existence of AE and SAE in Section~\ref{sec:existence of equilibrium}.}
%{\textcolor{red}{The above sentence seems redundant to me!}}

Next, we discuss the nature of the group power function $P_G(\mathbf{r}, \mathbf{x})$. Fundamentally, we are interested in power functions that (1) increase in the resources available to the group, and (2) decrease as the group members are more ``spread out.''
Specifically, we write $R_G = \sum_{i \in G} r_i$ for the total resources available to the group, and $D_G$ generically for a notion of how spread out the group is, referred to as the group's \emph{coverage}. Some natural examples are:

\begin{enumerate}
	\item The max pairwise distance 
	%\begin{equation}\label{eq:DG_diameter}
	$D_G \!=\! \max_{i,j \in G} \Vert{x_i - x_j}\Vert_2$. 
	%\end{equation} 
	\item The \emph{volume} of the convex hull of $\{x_i|i \in G\}$.
	\item The \emph{surface area} of the convex hull of $\{x_i|i \in G\}$; this captures a notion of the ``border'' of the group.% that the group must be cognizant of.
\end{enumerate}

These concepts are used later in the paper (Section \ref{sec: Structural Properties}), when we explicitly consider utilities under the hedonic form $U_G(R_G, D_G)$.

\begin{assumption} \label{assp:group-to-individual}
	Every agent in the same group $k$ shares the same group utility\footnote{%
		While we treat the agent's utility as \emph{equal} to the group's, one can easily add an agent-specific term which does not depend on the group's utility. Since this term does not depend on the action profile, it would be irrelevant to the analysis of equilibrium outcomes and is thus not explicitly modeled here.}.
	That is, if $i \in G := G_{\sigma_i}(\mathbf{a})$, then 
	\begin{align}
	u_i(\mathbf{a}) & = U_G(\mathbf{r},\mathbf{x}).
	\label{eq:individual utility}
	\end{align}
\end{assumption}

This assumption applies in cases where the group's combined achievements are enjoyed equally by its members; it rules out scenarios in which some members of the group hold privileged status, or overall benefit more from the group's pooled resources. For notational convenience later, we define the power of the empty group as $P_{\emptyset}=-\infty$ (thus $U_{\emptyset}=-\infty$).

%%%%%%%%%%%%%%%%%%%%%%%%%%%%%%%%%%%%%%%%%%%%%%%%%%%%%%%%%%%%%%%%%%%%%%
%%%%%%%%%%%%%%%%%%%%% subsec: Equilibrium %%%%%%%%%%%%%%%%%%%%%%%%%%%%

\subsection{Equilibrium and Stability}

Next, we define the notions of equilibria we study.
At a high level, our goal is to capture stability against deviations by individuals or groups. Importantly, such deviations to another group are possible only when the group accepts these new member(s). Recall that by Assumption~\ref{assp:group-to-individual}, all members of a group have the same utility. Therefore, all group members will accept a new member(s) if and only if one of them does; in other words, all approvals are automatically unanimous. We call these {\em acceptance} equilibria. We define two types of such equilibria, depending on whether we are considering individual deviations or group deviations. 

%An individual (or a group of individuals) can only join a group if the utility of all agents already in the group does not decrease by her (them) joining the group. We therefore call these equilibria (strong) acceptance equilibria, which in words are strategy profiles under which no agent can benefit from joining a group that would accept her (them).}%We therefore consider acceptance equilibria as strategy profiles under which no agent can benefit from joining a group that would accept her.

%%%%%%%%%% AE %%%%%%%%%%%%%%%%
\subsubsection*{Acceptance Equilibrium (AE)}
A group affiliation profile $\mathbf{\sigma}^*$ (and its corresponding strategy profile $\mathbf{a}$) is an \emph{Acceptance Equilibrium} (AE) if and only if no agent can benefit from joining a group that would accept her. Formally, $\mathbf{a}$ is an AE if and only if for all agents $i$ (writing $G = G_{\sigma^*_i}(\mathbf{a})$) and any group $G' = G_k(\mathbf{a})$ (including $G' = \emptyset$) with $k \neq \sigma^*_i$, at least one of the following two inequalities holds: \begin{align}
U_G(\mathbf{r}, \mathbf{x})  & \geq U_{G' \cup \{i\}}(\mathbf{r}, \mathbf{x}) \label{eq:ae1} \\
U_{G'}(\mathbf{r}, \mathbf{x}) & > U_{G' \cup \{i\}}(\mathbf{r}, \mathbf{x}). \label{eq:ae2}
\end{align}
The first inequality states that agent $i$ is weakly better off in her current group than by joining $G'$ (and hence would prefer not to deviate);
the second inequality states that group $G'$ is better off without having agent $i$ join, and hence prefers not to accept $i$.
By including $G' = \emptyset$ and recalling that $U_{\emptyset} = -\infty$, we also capture that $i$ would not prefer to deviate to being in a group by herself.

%%%%%%%%% SAE %%%%%%%%%%
\subsubsection*{Strong Acceptance Equilibrium (SAE)}
A group affiliation profile $\mathbf{\sigma}^*$ (and its corresponding strategy profile $\mathbf{a}$) is a \emph{Strong Acceptance Equilibrium (SAE)} if no \emph{subset of agents} from the same group can be better off by deviating together to another group that would accept them.
Formally, $\mathbf{a}$ is an SAE if and only if for every pair of groups $G = G_k(\mathbf{a}), G' = G_{k'}(\mathbf{a})$ (again, allowing for $G' = \emptyset$) and every subset $S \subseteq G$ of agents, at least one of the following two inequalities holds:
\begin{align}
U_G(\mathbf{r}, \mathbf{x}) \geq U_{G' \cup S} (\mathbf{r}, \mathbf{x}),  \label{eq:SAS1}\\
U_{G'}(\mathbf{r}, \mathbf{x}) > U_{G' \cup S} (\mathbf{r}, \mathbf{x}). \label{eq:SAS2}
\end{align}
The first inequality states that agents in $S$ weakly prefer staying in $G$ over deviating to join $G'$, the second that agents in $G'$ prefer not to accept the additional members $S$.
Again, by including $G'=\emptyset$, we capture that the members of $S$ do not prefer to form a new group.
Notice that it was possible to express this condition concisely by heavily exploiting Assumption~\ref{assp:group-to-individual}, namely, that all members of $S$ will obtain the \emph{same} utility in $G$, and also as new parts of $G'$.

\subsection{States and the ordering of groups and states} \label{sec:ordering}
We will consider \emph{dynamics} in which agents change their group affiliations and the dynamics' convergence to equilibria. For that reason, we also think of action profiles as \emph{states}, and deviations as transitions between these states. We now describe a specific way to order these states.
%For our proof of convergence of such dynamics, it will be very useful to order these states in a particular way, and we next describe this ordering.

%%%%%%%%%%%%%%%%%%%%% Group rank and relabeling %%%%%%%%%%%%%%%%%%%%%
Given a state $\mathbf{a}$, we order/label the resulting groups in the partition by non-decreasing group utility, breaking ties arbitrarily, e.g., lexicographically by some description of the group's membership.
Expressed formally, with $m = m(\mathbf{a})$ denoting the number of groups, we index the groups $G_1, G_2, \ldots, G_m$ such that 
$
U_{G_1} (\mathbf{r}, \mathbf{x}) \geq U_{G_2} (\mathbf{r}, \mathbf{x})
\geq \cdots \geq U_{G_m}(\mathbf{r}, \mathbf{x})
$, which implies 
$
P_{G_1} (\mathbf{r}, \mathbf{x}) \geq P_{G_2} (\mathbf{r}, \mathbf{x})
\geq \cdots \geq P_{G_m}(\mathbf{r}, \mathbf{x}). 
$
Similarly, we can index the agents $1,2, \ldots, n$ such that
$
u_1 (\mathbf{r}, \mathbf{x}) \geq u_2 (\mathbf{r}, \mathbf{x})
\geq \cdots \geq u_n (\mathbf{r}, \mathbf{x}). 
$
Since agents from the same group have the same utility, we label them consecutively. 
Notice that the indexing of groups (resp. agents) for two different states $\mathbf{a}, \mathbf{a'}$ may be very different, even when the resulting partitions share common groups (resp. agents). In particular, this observation is relevant when $\mathbf{a'}$ was obtained from $\mathbf{a}$ by the deviation of one agent or a subset of agents: the fact that the utilities of the affected groups (resp. agents) may have changed could result in a different ordering.

With each state $\mathbf{a}$, based on the above ordering, we associate a vector
$\mathbf{\Psi} (\mathbf{a}) = [u_1 (\mathbf{r}, \mathbf{x}), u_2 (\mathbf{r}, \mathbf{x}) , \ldots, u_n (\mathbf{r}, \mathbf{x})]$.
The vector $\mathbf{\Psi} (\mathbf{a})$ associated with each state $\mathbf{a}$ allows us to order the states lexicographically: a state $\mathbf{a}$ is ranked higher than $\mathbf{a'}$ in lexicographical order iff there exists an index $k \in \{1, \ldots, n\}$ such that: 
\begin{equation}
\psi_{k} (\mathbf{a})  > \psi_{k}(\mathbf{a'})~, 
\psi_{k'}(\mathbf{a}) = \psi_{k'}(\mathbf{a'}), ~\text{ for all } k' < k.
\label{eq:order2}
\end{equation}
In this case, we write $\mathbf{\Psi} (\mathbf{a}) \succ \mathbf{\Psi} (\mathbf{a'})$. Notice that it is possible to have $\mathbf{a} \neq \mathbf{a'}$ while $\mathbf{\Psi} (\mathbf{a}) = \mathbf{\Psi} (\mathbf{a'})$.
%Notice that when $\mathbf{a} = \mathbf{a'}$, i.e., the two vectors are equal, we consider an arbitrary but fixed tie-breaking; this will be inconsequential in our analysis, and allows us to define a strict lexicographical ordering for all states. We use $\mathbf{a}^{[k]}$ to denote the $k$-th state in this lexicographic order, so that
%$
%  \mathbf{\Psi} (\mathbf{a}^{[1]}) \succ \mathbf{\Psi} (\mathbf{a}^{[2]})
%  \succ \cdots. 
%$

%%%%%%%%%%%%%%%%%%%%%%%%%%%%%%%%%%%%%%%%%%%%%%%%%%%%%%%%%%%%%%%%%%%%%%%%
%%%%%%%%%%%%%%%%%%%%%%% SEC3: existence of equilibrium %%%%%%%%%%%%%%%%%
%%%%%%%%%%%%%%%%%%%%%%%%%%%%%%%%%%%%%%%%%%%%%%%%%%%%%%%%%%%%%%%%%%%%%%%%

\section{Existence of Equilibria and Convergence}
\label{sec:existence of equilibrium}

In this section, we prove that every instance of the GFG has at least one SAE and thus at least one AE (since any SAE is an AE). We in fact establish a much stronger fact: that any update dynamics under which agents (resp., subsets of agents who are currently in the same group) always strictly improve their utility converges to an AE (resp., SAE).

\subsection{Dynamics and Convergence}
Algorithm~\ref{al:BR for AE} captures a generic asynchronous improvement update algorithm. We will show that this algorithm converges to an SAE. Lemma~\ref{lemma:lexicographic-improvement} shows that every state transition corresponding to an improvement by a subset of agents from an existing group results in a higher-ranked new state according to the lexicographical ordering. Thus, using lexicographic rank as a potential function immediately implies Theorem~\ref{theo:best-response converge}, which shows that the dynamics converges to an SAE; in particular, this proves the existence of an SAE, and thus an AE. 
%%%%%%%%%%%%%%%%%% Algorithm 1  %%%%%%%%%%%%%%%%%%%%%%%%%%%%
%%%%%%%%%%%%%%%%%% best-response %%%%%%%%%%%%%%%%%%%%%%%%%%%

\begin{algorithm}
	\SetAlgoLined
	\LinesNumbered
	\SetKwInOut{Input}{Input}
	\SetKwRepeat{Do}{do}{while}
	\Input{$\mathbf{\sigma}^{(0)}$, a partition of $\mathcal{N}$}
	$t \leftarrow 0$\;
	\Do{there exists a set $S$ with $\text{IR}_S \neq \emptyset$}{\ForAll{$S \subset \mathcal{N}$}{
			\uIf{$S\subset G_k$ for some $k$}{$\text{IR}_S \leftarrow \{k'\,|\,U_{G_{k'} \cup S}(\mathbf{r}, \mathbf{x}) > U_{G_{\mathbf{\sigma}^{(t)}_S}}(\mathbf{r}, \mathbf{x}) \allowbreak \text{ and } U_{G_{k'} \cup S}(\mathbf{r}, \mathbf{x}) \geq U_{G_{k'}}(\mathbf{r}, \mathbf{x})\}$}
			\Else{$\text{IR}_S \leftarrow \emptyset$}
		}
		\If{there exists a set $S$ with $\text{IR}_S \neq \emptyset$}{
			Let $S$ be arbitrary such that $\text{IR}_S \neq \emptyset$\;
			Let $k' \in \text{IR}_S$ be arbitrary\;
			Obtain $\mathbf{\sigma}^{(t+1)}$ from $\mathbf{\sigma}^{(t)}$ by updating the group membership of agents in $S$ to $k'$, and leaving other memberships unchanged\;
		}
		$t \leftarrow t+1$\;
	}
	\caption{Asynchronous Improvement Update}
	\label{al:BR for AE}
\end{algorithm}

In Algorithm~\ref{al:BR for AE}, notice that $\text{IR}_S$ is the set of strictly improving responses for agents in $S$ (singleton groups are also included), i.e., the set of groups that strictly improve the agents' utility over their current utility, and would accept these agents as members. Also, note that $\text{IR}_S$ is only defined for subsets of agents that are currently in the same group. Since both the improving subset of agents and the specific new group are chosen arbitrarily, Algorithm~\ref{al:BR for AE} captures any dynamics in which subsets of agents who are currently in the same group always change strategies to improve their utility. Restricting $S$ in Algorithm~\ref{al:BR for AE} to be singletons, we obtain the same results for dynamics in which individuals change their strategies to improve their utility.

%%%%%%%%%%%%% Lemma 1    %%%%%%%%%%%%%%%%%%%%%%%%%%%%%%%%%%%%%%
\begin{lemma} \label{lemma:lexicographic-improvement}
	In Algorithm~\ref{al:BR for AE}, for each iteration $t$, the new state $\mathbf{a}^{(t+1)}$ (corresponding to $\mathbf{\sigma}^{(t+1)}$) strictly precedes the state $\mathbf{a}^{(t)}$ (corresponding to $\mathbf{\sigma}^{(t)}$) in the lexicographical order of states. That is, 
	\begin{align} \label{eq:lemma1-3}
	\mathbf{\Psi} (\mathbf{a}^{(t+1)}) & \succ \mathbf{\Psi} (\mathbf{a}^{(t)}). 
	\end{align}
\end{lemma}

\begin{proof}
	When a set of agents $S$ deviates from $G_k$ to another group $G_{k'}$, the only groups whose power may be affected by this are $G_k$ and $G_{k'}$. Let $\Tilde{k}$ and $\Tilde{k}'$ denote the indices of these two groups after the deviation, i.e., $G_{\Tilde{k}}:=G_k\setminus S$ and $G_{\Tilde{k}'}:=G_{k'}\cup S$.
	Let $\hat{k} \in \{k, k',\Tilde{k},\Tilde{k}'\}$ be the index of the group with the highest power among the groups (before/after deviation). The utility of groups $G_1, \ldots, G_{\hat{k}-1}$ is the same before and after the deviation; hence, the corresponding elements in $\mathbf{\Psi} (\mathbf{a}^{(t)})$ and $\mathbf{\Psi} (\mathbf{a}^{(t+1)})$ are equal. By the definition of $\text{IR}_S$, we have $U_{G_{k'} \cup S}(\mathbf{r}, \mathbf{x}) > U_{G_{k}}(\mathbf{r}, \mathbf{x})$ and $U_{G_{k'} \cup S}(\mathbf{r}, \mathbf{x}) \geq U_{G_{k'}}(\mathbf{r}, \mathbf{x})$. Notice that the same inequalities hold for the power of groups involved. Hence, $\hat{k}$ is either $\Tilde{k}$ or $\Tilde{k}'$. In the first case, the utility of all agents in $G_k\setminus S$ has increased after the deviation. In the latter case, the utility of all agents in $S$ has increased.\footnote{It is possible that $G_{k'} \cup S$ moves further up in the utility ranking, in which case a higher utility occurs even earlier. But this only helps the argument.} Therefore, the new vector $\mathbf{\Psi}(\mathbf{\sigma}^{(t+1)})$ ranks lexicographically before $\mathbf{\Psi}(\mathbf{\sigma}^{(t)})$.
\end{proof}

%%%%%%%%%%%%%%%%% theo: Convergence of best-response %%%%%%%%%%
\begin{theo} \label{theo:best-response converge}
	Algorithm~\ref{al:BR for AE} converges to an SAE in a finite number of steps.
	In particular, an SAE always exists, and thus an AE always exists.
\end{theo}

\begin{proof}
	By Lemma~\ref{lemma:lexicographic-improvement}, each iteration  results in a new state with higher $\mathbf{\Psi}(\mathbf{a})$, which is therefore ranked higher. When the dynamics terminates, no subset of agents $S$ from any existing groups has any improving updates that involve simultaneously changing their group affiliation. This exactly captures the equilibrium conditions \eqref{eq:SAS1} and \eqref{eq:SAS2}.
	Since there is only a finite number of states, the algorithm must converge to an SAE in finite time. 
\end{proof}
%By inspecting the proof of existence, it can be seen that the preceding results also hold when we relax Assumption~\ref{assp:group-depends-on-members} to allow a group's utility to depend on members of groups with higher utility. 
%A detailed discussion can be found in the Supplementary Material.

%%%%%%%%%%%%%%%%%% Algorithm 2 %%%%%%%%%%%%%%%%%%%%%%%%%%%%
%%%%%%%%%%%%%%%%%% find a SAE %%%%%%%%%%%%%%%%%%%%%%%%%%%%
We next introduce a different type of algorithm that will find an SAE which
contains the most powerful group among all possible subsets of $\mathcal{N}$. We call such an SAE a PSAE.
\begin{algorithm}
	\SetAlgoLined
	\LinesNumbered
	\SetKwInOut{Input}{Input}
	\SetKwRepeat{Do}{do}{while}
	\Input{$\mathcal{N}_1=\{1,2,\dots,n \}$}
	$k\leftarrow 1$\;
	\While{$\mathcal{N}_k \neq \emptyset$}{
		Pick $G_k \in \operatorname*{argmax}_{G \subseteq \mathcal{N}_k} U_G(\mathbf{r},\mathbf{x})$ which has the largest number of agents (otherwise, break ties arbitrarily)\;
		Set $\sigma_i = k$ for all $i \in G_k$\;
		$\mathcal{N}_{k+1} \leftarrow \mathcal{N}_{k} \setminus{G_{k}}$\;
		$k\leftarrow k+1$\;
	} 	
	\caption{Constructing a PSAE} 
	\label{al:SAE}
\end{algorithm}

When there are multiple groups with the same number of agents tied for the largest utility at iteration $k$ of Algorithm~\ref{al:SAE}, the arbitrary tie-breaker leads to possibly different less powerful groups for the remaining iterations, resulting in different values for the vector $\mathbf{\Psi} (\cdot)$.

Notice that if the tie-breaking rule is not invoked in any iteration, then the resulting PSAE is one of the highest-ranked states according to the lexicographical ordering. Also, note that all highest-ranked states are SAEs by Lemma~\ref{lemma:lexicographic-improvement}.

%%%%%%%%%%%%%%%%%%%%%%%%% SAE existence %%%%%%%%%%%%
\begin{theo}\label{theo:Algorithm2 converges}
	Algorithm~\ref{al:SAE} construct a PSAE.
	%thus an SAE with a most powerful group always exists.}
\end{theo}

\begin{proof}
	Notice that $n = |\mathcal{N}_1| > |\mathcal{N}_2| > \cdots$; hence, the algorithm stops after some finite number of iterations. Assume that the algorithm stops after $K \leq n$ iterations. We claim that for each $k \in \{1,2,\dots,K\}$, the individuals in group $G_k$ cannot join $G_{k'}$ with $k' < k$, and they have no incentive to join $G_{k'}$ with $k' > k$.
	
	%At iteration $k$ of the algorithm, 
	In the $k$th iteration, the individuals in group $G_k$ have the maximum possible group utility among all subsets of $\mathcal{N}_k$. Hence, no subset of $\mathcal{N}_k\setminus G_k$ can join $G_k$ and increase its utility. Similarly, individuals in $G_k$ do not have any incentive to deviate to any other subset of $\mathcal{N}_k$. Hence, the claim follows.	
\end{proof}	

%%%%%%%%%%%%%%%% polynomial time for convergence %%%%%%%%%%%%%%%%
\begin{prop}\label{prop:complexity}
	The complexity of Algorithm~\ref{al:SAE} is $\mathcal{O}(n^3)$ for group power functions of the form  $P_G(R_G, D_G)$, where $D_G = \max_{i,j \in G} \Vert{x_i - x_j}\Vert_2$.
\end{prop}

\begin{proof}
	Since $D_G$ ($|G| > 2$) only depends on the diameter of the point set, any individual that can join the group $G$ without changing its diameter is approved by all agents in group $G$. 
	
	By definition, $G_k$ has the maximum group utility among all subsets of $\mathcal{N}_k$ at iteration $k$.
	%\sout{Moreover, for any $i\in \mathcal{N}_k \setminus G_k$, the diameter of the convex hull of $G = \{i\} \cup G_k$ is strictly larger than $D_{G_k}$; otherwise, the group utility of $G$ would be higher than $G_k$.}
	Hence, $G_k$ is either a singleton group, or it includes everyone inside the ball given by its diameter $D_{G_k}$. As a result,
	%instead of searching all subsets of $\mathcal{N}_k$ to find $G_k$,
	we only need to check singleton groups, and the groups inside the ball given by any pair of individuals. In particular, the number of candidates for $G_k$ is $\mathcal{N}_k + \binom{\mathcal{N}_k}{2}$. Hence, the maximum complexity of the algorithm is given as follows:
	\[
	\sum_{k} \binom{\mathcal{N}_k}{2} + \mathcal{N}_k \leq n\cdot \binom{n}{2} + n = \mathcal{O}(n^3).
	\qedhere
	\popQED
	\]
\end{proof}

\subsection{Non-Uniqueness of Equilibria}
While Theorem~\ref{theo:best-response converge} and Theorem~\ref{theo:Algorithm2 converges} guarantee the existence of an SAE, the SAE may not be unique. This also implies that there may not be a unique AE. We show this in the following example.

%%%%%%%%%%%%%%%% Example: SAE not unique %%%%%%%%%%%%%%%%%%%%%%%%%
\begin{exmp}\label{eg:SAE not unique}
	Consider three heterogeneous agents on the line. Their locations are $x_1 = 0$, $x_2 = 0.6$, and $x_3=1.2$, and their resources are $r_1 = 2$, $r_2 = 1$, and $r_3 = 2$. 
	The group utility function is $U_G = \left(\sum_{i \in G} r_i\right) / \left(1 + \max_{i,j \in G} \lvert {x_i - x_j}\rvert\right)$.
	
	In this example, simple calculations show that a single large group and groups of isolated individuals are both SAEs. On the other hand, for the three partitions of $\{1,2,3\}$ into two sets, there is always an agent who prefers to deviate, resulting in either the partition $\{\{1\}, \{2\}, \{3\}\}$ or $\{\{1,2,3\}\}$.
\end{exmp}

%\textcolor{red}{\sout{Although a PSAE exists by Theorem~\ref{theo:Algorithm2 converges}, it may also not be unique, which also shows by the above example where the SAE is also a PSAE.}}

%%%%%%%%%%%%%%%%%%%%%%%%%%%%%%%%%%%%%%%%%%%%%%%%%%%%%%%%%%%%%%%%%%%%%%%%
%%%%%%%%%%%%%%%%%%%%%%% SEC4: Structural Properties %%%%%%%%%%%%%%%%%%%%
%%%%%%%%%%%%%%%%%%%%%%%%%%%%%%%%%%%%%%%%%%%%%%%%%%%%%%%%%%%%%%%%%%%%%%%%

\section{Structural Properties}\label{sec: Structural Properties}
Having established the existence of equilibria, we now turn to their properties.
In particular, we are interested in the combinatorial structure of overlap between different groups' ''territories.'' We formally define the following:

\begin{defi} \label{def:territory}
	The \emph{territory} of a group $G$, denoted by $X_G$, is the convex hull of its members' locations:
	\begin{align*}
	X_G & = \{\sum_{i \in G} a_i x_i|\sum_{i \in G} a_i = 1 \text{ and } a_i \geq 0 \text{ for all } i\}~. 
	\end{align*}
	\emph{Individuals inside the territory} of group $G$, denoted by $T_G$, are defined as
	$T_G = \{i\in\mathcal{N}|x_i \in X_G\}$.
\end{defi}

It is possible that $i \in T_G$ even when $i \notin G$; in fact, such structures are of particular interest to us. For the remainder of this section, we will adopt the following two additional assumptions. 

\begin{assumption}\label{asp:hull}
	The group coverage function $D_G$ depends only on the group's territory in the following sense: if $i \in T_G$, then $D_{G\cup\{i\}} = D_G$.
\end{assumption} 
In Section~\ref{sec:model}, we described the maximum pairwise distance, the volume of the convex hull, and the surface of the convex hull as three natural examples of group coverage.
Notice that all these examples satisfy Assumption~\ref{asp:hull}. 
\begin{assumption}\label{asp:utility}
	The group utility is of the hedonic form 
	\begin{align} \label{eq:utility-general2}
	U_{G} & = f(R_{G}, D_{G}), 
	\end{align}
	where the function $f$ is strictly increasing in $R_G$ and strictly decreasing in $D_G$. 
\end{assumption}

\subsection{Types of Structures in AE} 
We begin by defining four types of overlap (or lack thereof) between pairs of groups.
%%%%%%%%% Structure 1: non-overlapping %%%%%%%%%%%%%%%%%%%
\begin{defi} \label{def:group-overlaps}
	Let $G, G'$ be two groups.
	\begin{itemize}
		\item $G$ and $G'$ are \emph{non-overlapping} if their territories are disjoint, i.e., $X_G \cap X_{G'} = \emptyset$.
		\item We say that $G$ \emph{encroaches on} $G'$ if $G$ has at least one member in the territory of $G'$, i.e., $G \cap T_{G'} \neq \emptyset$.
		\item $G$ and $G'$ are \emph{mutually encroaching} if $G \cap T_{G'} \neq \emptyset $ and $G' \cap T_{G} \neq \emptyset$.
		\item $G$ and $G'$ are \emph{nested} if all members of $G$ are located within the territory of $G'$, i.e., $G \subseteq T_{G'}$.
	\end{itemize}
\end{defi}

The four types of relationships are illustrated in Figure~\ref{fig:non-overlapping}. Notice that non-mutual encroachment between groups in one dimension can only occur when one group is nested inside the other; however, in higher dimension, non-mutual encroachment can happen even when the groups are not nested.

%%%%%%%%% fig: non-overlapping %%%%%%%%
\begin{figure}[ht!]
	\centering
	\includegraphics[width=0.45\columnwidth]{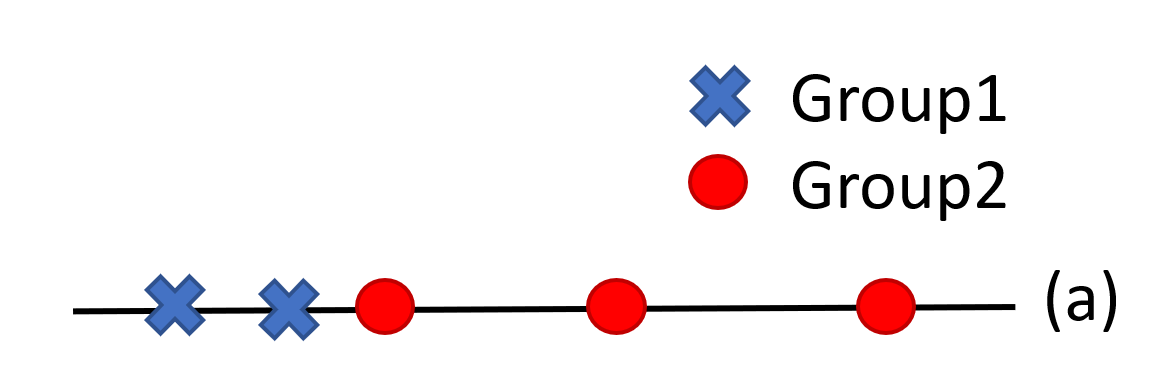}
	\hspace{6pt} 
	\includegraphics[width=0.45\columnwidth]{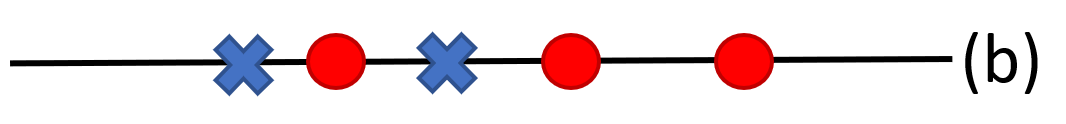}
	
	\includegraphics[width=0.45\columnwidth]{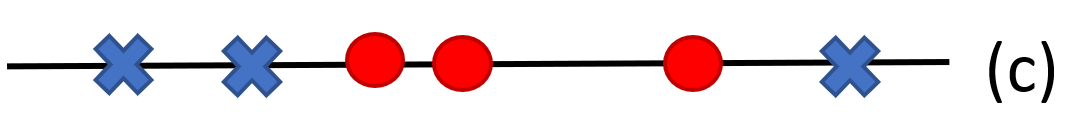}
	\hspace{6pt} 
	\includegraphics[width=0.45\columnwidth]{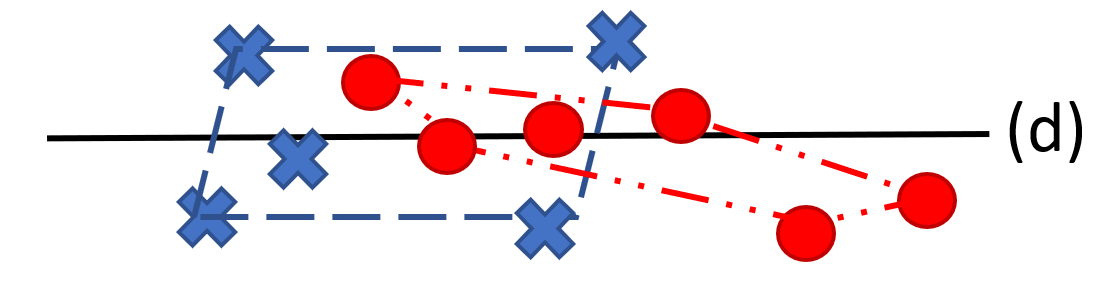}
	\vspace{-4pt} 
	\caption{Two groups that are (a) non-overlapping, (b) mutually-encroaching, (c) nested, in 1D, and (d) one encroaching on another in 2D. \label{fig:non-overlapping}}
\end{figure}

The following proposition states that mutual encroachment cannot occur in any AE.

\begin{prop} \label{prop:struc2not-NE/AE}
	There is no AE in which two groups are mutually encroaching.
\end{prop}

\begin{proof}
	Assume that there is an AE with group affiliation profile $\mathbf{\sigma}^*$ in which there are two groups $G$ and $G'$ such that there exist agents $i \in G \cap T_{G'}$ and $j \in G' \cap T_{G}$. Without loss of generality, assume that $f(R_{G'}, D_{G'}) \geq f(R_{G}, D_{G})$.
	
	Since $i \in T_{G_{k'}}$, by Assumption~\ref{asp:hull}, we have $D_{G'} = D_{G' \cup \{i\}}$. Hence, $f(R_{G' \cup \{i\}}, D_{G' \cup \{i\}}) > f(R_{G'}, D_{G'})$ since $R_{G' \cup \{i\}} > R_{G'}$. Therefore, $f(R_{G' \cup \{i\}}, D_{G' \cup \{i\}}) > f(R_{G}, D_{G})$, so agent $i$ has an acceptable beneficial deviation to group $G'$. This contradicts the assumption that $\mathbf{\sigma}^*$ is an AE.
\end{proof}

The other types of group interactions may exist in an AE. The following are some examples:
\begin{itemize}
	\item Non-overlapping groups: this occurs when groups are far from each other; thus, no one has any incentive to deviate.
	\item Nested Structure: this occurs when a group with high resources is located within a much weaker group. Agents of the weaker group cannot deviate, as they will reduce the other group's utility by enlarging its territory.
	\item Non-nested one-way encroaching structure: two groups may overlap as shown in Figure~\ref{fig:non-overlapping}(d), where if the red group has higher utility, its agents may not want to deviate to the blue group.
\end{itemize}

\subsection{ Encroachment Structures in AE} 
We are now ready to attack our main question: what types of overlaps can occur globally? In particular, we are interested in what types of encroachment relations can occur between the groups at equilibrium.
To characterize these relations, we define the encroachment graph.

\begin{defi} \label{def:encroachment-graph}
	Given an AE $\mathbf{a}^{*}$ and its group partition $G_1, \ldots, G_m$, we define the directed \emph{encroachment graph} $\mathcal{G}(\mathbf{a}^{*})$ as follows:
	the $m$ nodes $V(\mathbf{a}^{*})$ are the groups $G_1, \ldots, G_m$, and there is a directed edge from $G_i$ to $G_j$ if and only if $G_i$ encroaches on $G_j$.
\end{defi}

\begin{prop}\label{theo:directed graph}
	For every AE $\mathbf{a}^{*}$, the encroachment graph $\mathcal{G}(\mathbf{a}^{*})$ is acyclic.
\end{prop}

\begin{proof}
	Let $G, G'$ be two arbitrary groups such that $\mathcal{G}(\mathbf{a}^{*})$ contains the directed edge $(G,G')$, i.e., $G$ encroaches on $G'$.
	By definition, this means that $G \cap T_{G'} \neq \emptyset$.
	Let $i \in G \cap T_{G'}$ be an arbitrary agent.
	Because the convex hull of $G'$ is the same as the convex hull of $G' \cup \{i\}$, we obtain that $D_{G' \cup \{i\}} = D_{G'}$, by Assumption~\ref{asp:hull}. 
	
	Because $R_{G' \cup \{i\}} > R_{G'}$, we obtain that $U_{G' \cup \{i\}} > U_{G'}$, i.e., the group $G'$ would strictly prefer to accept $i$.
	
	The fact that $\mathbf{a}^{*}$ is an AE therefore implies that $i$ does not want to deviate and join $G'$. This means that $U_G \geq U_{G' \cup \{i\}} > U_{G'}$. We have thus shown that whenever there is a directed edge from $G$ to $G'$, the utility of $G$ is strictly higher than that of $G'$. As a result, $\mathcal{G}(\mathbf{a}^{*})$ cannot contain any cycles.
\end{proof}

Next, we show that directed acyclic graphs (DAGs) precisely characterize encroachment relationships. 
We show that given any DAG $\mathcal{G}$, there exists a system of agents, locations, resources, utility functions, and an AE $\mathbf{a}^{*}$ such that $\mathcal{G}(\mathbf{a}^{*}) = \mathcal{G}$.
We do this by explicitly constructing the system, as described below.

In fact, we show that such a construction is always possible in any fixed dimension $d \geq 2$, and for very general classes of group utility functions.
Specifically, we will show this for any group utility function $U_G = f(R_G,D_G)$, whenever
\begin{enumerate}
	\item $D_G$ is a strictly increasing function of $X_G$, i.e., for any $G, G' \subseteq \mathcal{N}$ with $X_{G} \subsetneq X_{G'}$, we have $D_{G} < D_{G'}$.
	\item $f(R,D)$ grows unboundedly in $R$, and 
	\item for all $r,D,\delta > 0$,
	\begin{align}
	\limsup_{R\to\infty} \frac{f(R,D)}{f(R+r,D+\delta)} & > 1. \label{eq:condf}
	\end{align}                                                         
\end{enumerate}
We will show in Proposition~\ref{prop:function-condition} below that a wide class of functions satisfy these conditions.

Consider the following construction procedure, which is given an arbitrary DAG $\mathcal{G}$ with $m$ nodes as input: 
\begin{enumerate}
	\item Let $\epsilon$ be a very small constant.
	\item Let $v_1, \ldots, v_m$ be a topological sorting of nodes of $\mathcal{G}$.
	\item Let $\{\ell_1,\ell_2,\ldots,\ell_{m}\}$ be  $m$ pairwise intersecting line segments in $\mathbb{R}^d$ of the same length. (Each line segment corresponds to a group.)
	\item For each line segment, locate $d$ agents in an $\epsilon > 0$ neighborhood of each its endpoints, so that the convex hull of individuals in the $\epsilon$-neighborhood of each line segment is non-degenerate, and the convex hulls of the agents for different line segments are identical up to translation and rotation.
	\item For every pair $i, j$ such that $(v_i,v_j) \in E(\mathcal{G})$, locate an agent at the point where $\ell_i$ and $\ell_j$ intersect.
	\item Inductively, for every $k > 1$, define $G_k$ to be the set of all agents that are in the $\epsilon$-neighborhood of $\ell_k$ and are not in $G_{k'}$ for any $k' < k$.
	($\epsilon$ is chosen small enough so that all of the designated nodes near the endpoints of $\ell_k$ will be in $G_k$.)
	\item Add individuals to each group (without changing its convex hull) so that $|G_1| > |G_2| > \cdots > |G_m|$.
	\item Assign the resources $r_k$ to the agents of $G_k$ for $k \in \{1,\ldots,m\}$ so that the resulting configuration is an AE.
\end{enumerate}

The key step in the proof is to show that we can define the resources $(r_1, r_2, \ldots, r_m)$ so that individuals in $G_k$ do not have an incentive to join $G_{k'}$ for $k' > k$, and so that they are not allowed to join $G_{k'}$ for $k'<k$.

\begin{lemma} \label{lemma:can-pick-resources}
	In the preceding construction, under Assumptions~\ref{asp:hull} and \ref{asp:utility}, we can pick $(r_1,r_2,\ldots,r_m)$ so that the resulting configuration is an AE.
\end{lemma}

\begin{proof}
	Let $\mathbf{a}^{*}$ be the partition produced by the given procedure, and let $n_i = |G_i|$ for all $i$.
	For each $k$, let
	\begin{align*}
	\delta_k & =  \min_{k' > k, j \in G_{k'}}(D_{G_k \cup \{j\}} - D_{G_k})
	\end{align*}
	be the smallest change to $D_{G_k}$ caused by any agent $j \in G_{k'}$ for $k' >k$ joining $G_k$.
	
	Since $D_G$ is a strictly increasing function of $X_G$ by assumption, and each $j \in G_{k'}$ for $k' > k$ lies outside $X_{G_k}$ 
	by construction, $\delta_k > 0$ for all $k$.
	Define $\delta = \min_{k} \delta_k$.
	Because the convex hulls of all groups are identical up to rotation and translation, we have that $D_{G_1} = D_{G_2} = \cdots = D_{G_m}$; denote this common value by $D_G$.
	The following conditions are sufficient for the given group structures to form an AE:
	\begin{itemize}
		\item Individuals are not allowed to join a group with higher utility. A sufficient condition for this is that for all $k < m$, 
		\begin{align}
		f(r_k n_k + r_{k+1}, D_G + \delta) & <  f(r_k n_k, D_G). \label{ineq:1}
		\end{align}
		Notice that the right-hand side is the current utility of $G_k$, while the left-hand side is an upper bound on the utility of $G_k \cup \{j\}$ for any $j \in G_{k'}$ with $k' > k$. 
		This is because the monotonicity of $f$, along with $r_{k+1} \geq r_{k'}$ for all $k' > k$, ensures that it is enough to consider $r_{k+1}$ here.
		\item Individuals do not have any incentive to join a group with lower utility. A sufficient condition for this is that for all $k < k' \leq m$,
		\begin{align} \label{ineq:2}
		f(r_{k'} n_{k'} + r_k, D_G + \delta) & < f(r_k n_k, D_G)~. 
		\end{align}
		Here, the right-hand side is the current utility of agents in group $k$, while the left-hand side is the new utility they would experience if joining the group $k' > k$. 
	\end{itemize}
	We assign the values $r_1 > r_2 > \cdots > r_m$ iteratively, starting with $r_m =1$.
	Given the values of $r_{i}$ for $i > k$, we want to define $r_{k}$ so that Inequalities~\eqref{ineq:1} and \eqref{ineq:2} are simultaneously satisfied.
	Notice that \eqref{ineq:2} holds whenever $r_k \geq r_{k+1}$, because $n_{k'} < n_k$ for all $k' > k$ and $f(R,D)$ is increasing in $R$ and decreasing in $D$.
	
	Because $f$ satisfies \eqref{eq:condf}, applying it with the given $\delta$, $D = D_G$, and $r = r_{k+1}$ implies that there exists a large\footnote{The fact that $R$ can be made arbitrarily large is important in that it ensures that we will be able to choose $r_k \geq r_{k+1}$.} value of $R$ with $\frac{f(R,D_G)}{f(R+r_{k+1},D_G+\delta)} > 1$. Setting $r_k = R/n_k$ then implies that $f(r_k n_k + r_{k+1}, D_G + \delta) <  f(r_k n_k, D_G)$, i.e., \eqref{ineq:1}.
	This completes the iterative construction and thus the proof.
\end{proof}

In summary, we have proved the following theorem:
\begin{theo} \label{theo:all directed graph}
	Let $d \geq 2$ be any fixed dimension.
	Assume that $D_G$ is a strictly increasing function of $X_G$. Let $U_G$ satisfy Assumption \ref{asp:utility} and $f(R_G, D_G)$ satisfy condition~\eqref{eq:condf}. Given any DAG $\mathcal{G}$, there exists a group formation game in $d$ dimensions with group utility function $f$ and an AE $\mathbf{a}^{*}$, such that the encroachment graph of $\mathbf{a}^{*}$ is $\mathcal{G}$.
\end{theo}

Finally, we show that a wide class of natural utility functions satisfy the conditions of Theorem~\ref{theo:all directed graph}.

\begin{prop} \label{prop:function-condition}
	Let $f$ be of the form $f(R,D) = g(R)/h(D)$, and assume that $h(\cdot)$ is increasing, and that $g(\cdot)$ is increasing and sub-exponential, i.e.,
	$   \limsup_{R\to\infty} \frac{\log(g(R))}{R} = 0$. 
	Then $f(\cdot,\cdot)$ satisfies the condition \eqref{eq:condf}.
\end{prop}
\begin{proof}
	For $f$ of the separable form $f(R,D) = g(R)/h(D)$, we can rewrite \eqref{eq:condf} as:
	\begin{align}
	\limsup_{R\to\infty}\frac{g(R)}{g(R+r)}\frac {h( D+ \delta)}{h(D)} & > 1, \label{eq:con2}
	\end{align}
	Since ${h(D + \delta)}/{h(D)} > 1$, it is sufficient to show that
	$\limsup_{R\to\infty}\frac{g(R)}{g(R+r)} = 1$.
	We prove this by contradiction, and assume that $\limsup_{R\to\infty}{g(R)}/{g(R+r)} = \alpha < 1$. Hence, there exists some sufficiently large $R_0 > 0$ such that $g(R)/g(R+r) < (\alpha+1)/2$ for all $R > R_0$.
	This inequality can be rearranged to $g(R) < \frac{\alpha+1}{2} \cdot g(R+r)$.
	Let $k = \lfloor{R_0/r}\rfloor +1$, and consider the sequence $(g((k+i) r))_{i=0}^\infty$.
	Since $(k+i) r > R_0$ for all $i$, we have
	\begin{align*}
	g(k r) & < \frac{\alpha+1}{2} g((k+1) r)
	< \frac{(\alpha+1)^2}{2^2}g((k+2) r)
	< \cdots
	\end{align*}	
	Hence,
	$g((k+i) r) > \left(\frac{2}{\alpha+1}\right)^i g(k r)$
	for all $i>0$.
	Taking logarithms and dividing by $(k+i)$, we get
	\begin{align*}
	\frac{\log(g((k+i) r))}{(k+i)r}
	& >  \frac{i\cdot\log\left(\frac{2}{\alpha+1}\right) }{(k+i)r} + \frac{\log(g(k r))}{(k+i)r} .
	\end{align*}
	As $i\to\infty$, the right-hand side of the above inequality converges to $\frac{\log(2/(\alpha+1))}{r} > 0$, contradicting the assumption that $g$ is a sub-exponential function.
\end{proof}

Notice that subexponential growth in the \emph{resources} is a very weak condition that is easily satisfied. 

%%%%%%%%%%%%%%%%%%%%%%%%%%%%%%%%%%%%%%%%%%%%%%%%%%%%%%%%%%%%%%%%%%%%%%%%
%%%%%%%%%%%%%%%%%%% SEC5: Conclusion and Future Work %%%%%%%%%%%%%%%%%%%
%%%%%%%%%%%%%%%%%%%%%%%%%%%%%%%%%%%%%%%%%%%%%%%%%%%%%%%%%%%%%%%%%%%%%%%%
\section{Conclusion and Future Work}
We studied the existence and characteristics of different types of equilibria in a group formation game in which agents benefit from being part of a well-resourced and cohesive group.
We investigated two types of equilibria, AE and SAE, in this game and explored their various structural properties. In particular, we showed that each AE can be represented by a DAG, and conversely, that under mild conditions on the utility functions, every DAG can arise as an AE of a suitably defined instance of the group formation game.

While the assumption that $D_G$ is a function of only $X_G$ is natural, when it does not hold, some DAGs may not arise as equilibrium encroachment graphs; this would be a direction for future research. In addition, one can naturally extend the model to contain \emph{hierarchical} structure, i.e., groups of groups, whereby an agent obtains utility from groups at each level that she belongs to. A deeper analysis of such models may shed interesting light on the emergence and stability of hierarchical organizations. From an algorithmic viewpoint, it will be interesting to study what minimal changes a principal might perform on the locations or resources of agents in order to achieve the formation of a particular group structure.

%%%%%%%%%%%%%%%%%%%%%%%%%%%%%%%%%%%%%%%%%%%%%%%%%%%%%%%%%%%%%%%%%%%%%%%%%%%%%%%%

%%%%%%%%%%%%%%%%%%%%%%%%%%%%%%%%%%%%%%%%%%%%%%%%%%%%%%%%%%%%%%%%%%%%%%%%%%%%%%%%
%\section*{APPENDIX}


\section*{ACKNOWLEDGMENT}
%\textcolor{cyan}{
This work has been
%\dk{R: has been} 
supported by the NSF under grants CNS- 1939006, CNS-2012001, ATD-2027277, CCF 1934986 and by the ARO under contract W911NF1810208.
%}

%%%%%%%%%%%%%%%%%%%%%%%%%%%%%%%%%%%%%%%%%%%%%%%%%%%%%%%%%%%%%%%%%%%%%%%%%%%%%%%%

%\begin{thebibliography}{99}

%\bibitem{c1} G. O. Young, ÒSynthetic structure of industrial plastics (Book style with paper title and editor),Ó 	in Plastics, 2nd ed. vol. 3, J. Peters, Ed.  New York: McGraw-Hill, 1964, pp. 15Ð64.

%\end{thebibliography}
\nocite{*}
\bibliographystyle{plain}
\bibliography{myreference}

\end{document}